\documentclass{scrartcl}

\usepackage{tagging}
\usetag{arxiv}

\usepackage{xparse}
\usepackage[english]{babel}

\usepackage{csquotes}
\usepackage{enumitem}
\usepackage[noend]{algpseudocode}
\usepackage{color}
\usepackage{amsmath,amsthm}
\usepackage{thmtools}
\usepackage{stmaryrd}
\usepackage{mathtools}
\usepackage{tikz}
\usepackage{soul}
\usepackage[section]{placeins}

\tagged{arxiv,personal}{
	\usepackage[maxbibnames=99,style=alphabetic,backend=bibtex]{biblatex}
	\usepackage[noblocks]{authblk}
	\usepackage{amssymb}
	\addbibresource{mybibliography.bib}
}

\usepackage{hyperref}
\usepackage{algorithm}
\usepackage[nameinlink,capitalise]{cleveref}

\newtheorem{theorem}{Theorem}
\newtheorem{lemma}[theorem]{Lemma}

\newtheorem{definition}[theorem]{Definition}

\newtheorem*{example*}{Example}

\providecommand{\diedge}[1]{(#1)}

\newcommand{\fail}{\textbf{Fail}}

\newcommand{\RR}{\ensuremath{\mathcal{R}}}
\newcommand{\ww}{\ensuremath{\omega}}
\NewDocumentCommand{\dds}{m o}{%
	d[#1]\IfNoValueF{#2}{_{#2}}%
}

\def\xth/{\textsuperscript{th}}
\algtext*{EndWhile}
\algtext*{EndIf}
\algtext*{EndFor}
\algtext*{EndProcedure}

\makeatletter
\newcommand{\StatexIndent}[1][3]{%
 	\setlength\@tempdima{\algorithmicindent}%
 	\Statex\hskip\dimexpr#1\@tempdima\relax}
\makeatother
\algblockdefx{Pass}{EndPass}[2]{\textbf{p\ul{ass} #1}; input: #2}{}{}%
\algblockdefx{Passes}{EndPasses}[3]{\textbf{p\ul{asses} #1 to #2}; input: #3}{}{}%
\algblockdefx{ForPar}{EndForPar}[1]{\textbf{p\ul{arallel} for all} #1}{}{}%
\makeatletter
\ifthenelse{\equal{\ALG@noend}{t}}%
  {\algtext*{EndPass}}
  {}%
  \ifthenelse{\equal{\ALG@noend}{t}}%
  {\algtext*{EndPasses}}
  {}%
\ifthenelse{\equal{\ALG@noend}{t}}%
  {\algtext*{EndForPar}}
  {}%
\makeatother

\NewDocumentCommand{\q}{m m}{%
    q^{#1,#2}%
}
\NewDocumentCommand{\qt}{m m}{%
    t^{#1,#2}%
}
\NewDocumentCommand{\qx}{m m}{%
    x^{#1,#2}%
}
\NewDocumentCommand{\dg}{o m}{%
    \mathrm{dg}\IfNoValueF{#1}{_{#1}}(#2)%
}

\begin{document}
\tagged{pods}{\fancyhead{}}

\begin{taggedblock}{arxiv,personal}
\title{Approximately Counting Subgraphs in Data Streams}

\author{Hendrik Fichtenberger}
\affil{University of Vienna\\ Austria\\ \href{mailto:hendrik.fichtenberger@univie.ac.at}{hendrik.fichtenberger@univie.ac.at}}

\author{Pan Peng}
\affil{University of Science and Technology of China\\ China\\ \href{mailto:ppeng@ustc.edu.cn}{ppeng@ustc.edu.cn}}	

\date{}
\end{taggedblock}

\begin{taggedblock}{pods}
\title{Approximately Counting Subgraphs in Data Streams}

\author{Hendrik Fichtenberger}
\email{hendrik.fichtenberger@univie.ac.at}
\orcid{0000-0003-3246-5323}
\affiliation{%
  \institution{University of Vienna}
  \city{Vienna}
  \country{Austria}
}

\author{Pan Peng}
\email{ppeng@ustc.edu.cn}
\orcid{0000-0003-2700-5699}
\affiliation{%
	\institution{School of Computer Science and Technology\\University of Science and Technology of China}
  \country{Hefei, China}
}

\end{taggedblock}

\newcommand{\theabstract}{
	\begin{abstract}
		Estimating the number of subgraphs in data streams is a fundamental problem that has received great attention in the past decade. In this paper, we give improved streaming algorithms for approximately counting the number of occurrences of an arbitrary subgraph $H$, denoted $\# H$, when the input graph $G$ is represented as a stream of $m$ edges. To obtain our algorithms, we provide a generic transformation that converts constant-round sublinear-time graph algorithms in the query access model to constant-pass sublinear-space graph streaming algorithms. Using this transformation, we obtain the following results.
		\begin{itemize}
			\item We give a $3$-pass turnstile streaming algorithm for $(1\pm \epsilon)$-approximating $\# H$ in $\tilde{O}(\frac{m^{\rho(H)}}{\epsilon^2\cdot \# H})$ space, where $\rho(H)$ is the fractional edge-cover of $H$. This improves upon and generalizes a result of McGregor et al. [PODS 2016], who gave a $3$-pass insertion-only streaming algorithm for $(1\pm \epsilon)$-approximating the number $\# T$ of triangles in $\tilde{O}(\frac{m^{3/2}}{\epsilon^2\cdot \# T})$ space if the algorithm is given additional oracle access to the degrees.
			\item We provide a constant-pass streaming algorithm for $(1\pm \epsilon)$-approximating $\# K_r$ in $\tilde{O}(\frac{m\lambda^{r-2}}{\epsilon^2\cdot \# K_r})$ space for any $r\geq 3$, in a graph $G$ with degeneracy $\lambda$, where $K_r$ is a clique on $r$ vertices. This resolves a conjecture by Bera and  Seshadhri [PODS 2020].
		\end{itemize}
		More generally, our reduction relates the adaptivity of a query algorithm to the pass complexity of a corresponding streaming algorithm, and it is applicable to all algorithms in standard sublinear-time graph query models, e.g., the (augmented) general model.
	\end{abstract}
}

\begin{taggedblock}{pods}
\theabstract

\end{taggedblock}
\tagged{pods}{
}

\maketitle
\tagged{arxiv,personal}{\theabstract}

\usetag{mainpart}

\section{Introduction}
Estimating the number of occurrences of a small target graph (e.g., a triangle or a clique) in a large graph is a fundamental problem that has received great attention in many domains, including database theory, network science, data mining and theoretical computer science. For example, in database theory, it is closely related to the subgraph enumeration problem and the join-size estimation problem (see, e.g., \cite{assadi2019simple}). In network science, it has applications in estimating the transitivity coefficient and clustering coefficient of a social network (e.g., \cite{pavan2013counting}), and motif detection in biological networks (e.g., \cite{grochow2007network}).  

In this paper, we study this problem in the streaming setting. That is, we are given an $n$-vertex graph $G$ with $m$ edges that is represented as a stream of edge updates, and a (small) target graph $H$ (e.g., a triangle or a clique). Our goal is to estimate the number of occurrences of $H$ in $G$ by using as small space as possible in a few number of passes over the stream. 
Throughout the paper, we focus on the \emph{arbitrary-order model}, i.e., the order of the elements in the stream is arbitrary and may be adversarial. The baseline of graph streaming algorithms is the \emph{insertion-only} setting (also known as \emph{cash-register} setting), where the edges of $G$ are given one by one as a stream. When explicitly stated, we also consider the \emph{turnstile} setting, where the stream consists of insertions and deletions (similar to the model of dynamic algorithms). In the latter, the graph $G$ results from applying these insertions and deletions to an initially empty graph on $n$ vertices in the order as they are read from the stream. Each model is relevant for different types of applications: Multi-pass insertion-only algorithms allow to process very large graphs that do not fit into memory as entries in adjacency lists can be seen as insertions. Multi-pass turnstile algorithms can be applied even if a stream of insertions and deletions cannot be consolidated into an insertion-only stream of the final graph, e.g., because the stream is split into multiple substreams that cannot be joined for privacy reasons.

As we will discuss below, the special cases of $H$ being a triangle, a cycle or a clique have been widely studied. However, to the best of our knowledge, the only previous streaming algorithms for $(1\pm \epsilon)$-approximating the number of copies of an \emph{arbitrary} subgraph $H$ are the following:
\begin{enumerate}
        \item Kane et al. \cite{kane2012counting} gave a 1-pass turnstile algorithm that (i) uses  $\tilde{O}(\frac{(m\cdot \Delta(G))^{|E(H)|}}{\epsilon^2\cdot (\# H)^2})$ space for any subgraph $H$, where $\# H$ is the number of occurrences of $H$ in $G$ and $\Delta(G)$ is the maximum degree of $G$, or (ii) uses $\tilde{O}(\frac{m^{|E(H)|}}{\epsilon^2\cdot (\# H)^2})$ space if the minimum degree of $H$ is at least $2$. 
    \item Bera and Chakrabarti \cite{bera2017towards} gave a 2-pass algorithm with space $\tilde{O}(m^{\beta(H)}/(\epsilon^2 \# H))$, where $\beta(H)$ is the \emph{integral edge cover} number\footnote{The integral edge cover of $H$, denoted $\beta(H)$, is the cardinality of its smallest edge cover, where an edge cover of $H$ is a set of edges that covers all its vertices.  It is known that for an $r$-clique $K_r$, $\beta(K_r)=\lceil\frac{r}{2}\rceil$, and for a length-$r$ cycle $C_r$, $\beta(C_r)=\lceil\frac{r}{2}\rceil$. } of $H$.
    \item Assadi et al. \cite{assadi2019simple} gave a $C$-pass streaming algorithm with space complexity $\tilde{O}(\frac{m^{\rho(H)}}{\epsilon^2\cdot \# H})$, where $\rho(H)$ is the \emph{fractional edge cover} of $H$ (see \cref{def:edgecover}) and $C$ is some constant depending on $H$. We remark that $C$ is not explicitly specified in \cite{assadi2019simple}, but as far as we can see, a straightforward transformation of their sublinear-time algorithm to the insertion-only streaming setting gives that $C \geq \rho(H) \in \Omega(|V(H)|)$.\footnote{In \cite{assadi2019simple}, a so-called sampler tree of depth at least $\rho(H)$ is built top-down by querying the graph. To obtain a sufficiently small bound on the space, level $i$ of the tree must be fully constructed before level $i+1$ can be constructed. It seems necessary that the algorithm makes at least one full pass to construct a single level.} Furthermore, it is known that $\rho(H)\leq \beta(H)\leq |E(H)|$ and that $\#H \leq m^{\rho(H)}$ \cite{atserias2008size}, and thus the space complexity of the algorithm in \cite{assadi2019simple} is always no worse than the ones in \cite{kane2012counting,bera2017towards}.
\end{enumerate} 

Furthermore, it is known that 1-pass turnstile algorithms need at least $\tilde{\Omega}(m / (\# H) ^{1/\tau})$ space, where $\tau$ is the fractional vertex-cover of $H$ (analogous to \cref{def:edgecover}), even for bounded-degree graphs \cite{kallaugher2018sketching}.

There exist many streaming algorithms for the special cases of $H$ being an $r$-clique $K_r$, or a length-$r$ cycle $C_r$, for any constant $r\geq 3$. The performance guarantees of these algorithms are parameterized by various parameters of $G$, e.g., $\# H$, the maximum degree, the maximum number of triangles which share a single vertex (or an edge),  etc. In the following, we will mainly discuss the state-of-the-art results that are most relevant to our setting, i.e., those that provide $(1\pm \epsilon)$-approximations with space complexity parameterized just by $\# H$ (and $n,m,\epsilon$) in the arbitrary-order model. 

\paragraph{Triangles. ($\rho(C_3) = 3/2$)} Approximating the number of occurrences of a triangle $T$ has been studied in a long line of work \cite{bar2002reductions,ahn2012graph,bera2017towards,braverman2013hard,bulteau2016triangle,buriol2006counting,cormode2014second,jowhari2005new,mcgregor2016better,pavan2013counting,kolountzakis2012efficient,tsourakakis2009doulion,pagh2012colorful,jayaram2021optimal,kallaugher2017hybrid,manjunath2011approximate}. In one pass, Manjunath et al. \cite{manjunath2011approximate} gave one algorithm achieving  $\tilde{O}(\frac{m^{3}}{\epsilon^2\cdot (\#T)^2})$ space (in the turnstile model), which is nearly optimal as any 1-pass algorithm for this problem requires $\Omega(\frac{m^{3}}{(\#T)^2})$ space \cite{bulteau2016triangle}. In two passes,  McGregor, Vorotnikova and Vu \cite{mcgregor2016better} gave one algorithm using $\tilde{O}(\frac{m}{\epsilon^2\cdot \sqrt{\# T}})$ space (see also \cite{cormode2017second}). This is in contrast with a lower bound $\Omega(\min\{\frac{m}{\sqrt{\# T}}, \frac{m^{3/2}}{\# T}\})$ for any multi-pass algorithm by Bera and Chakrabarti \cite{bera2017towards}. In three passes, McGregor, Vorotnikova and Vu \cite{mcgregor2016better} gave one algorithm using $\tilde{O}(\frac{m^{3/2}}{\epsilon^2\cdot \# T})$ space, while \emph{their algorithm is assumed to have oracle access to vertex degrees}. In four passes, Bera and Chakrabarti \cite{bera2017towards} gave an algorithm using $\tilde{O}(\frac{m^{3/2}}{\epsilon^2\cdot \# T})$ space.

\paragraph{Cycles. ($\rho(C_r) = r/2$)} The case of counting a length-$r$ cycle $C_r$ (for some constant $r\geq 4$) has been studied in \cite{manjunath2011approximate,bera2017towards,mcgregor2020triangle}. In one pass, the turnstile algorithm in \cite{manjunath2011approximate} achieves $\tilde{O}(\frac{m^{r}}{\epsilon^2\cdot (\# C_r)^2})$ space, which is in contrast to a 1-pass space lower bound $\Omega(m^{r/2}/(\# C_r)^2)$ for even $r$ and $\Omega(m^r/(\# C_r)^2)$ for odd $r$ \cite{bera2017towards}. There exists an algorithm with space complexity $\tilde{O}(\frac{m^{r/2}}{\varepsilon^2\cdot\# C_r})$ using two passes for even $r$ and four passes for odd $r$ \cite{bera2017towards}.  %
In contrast, any multi-pass streaming algorithm requires $\Omega(m^{r/2}/\# C_r)$ space for even $r$ and $\Omega(\min\{\frac{m^{r/2}}{\# C_r},\frac{m}{(\# C_r)^{1/(r-1)}}\})$ space for odd $r$  \cite{bera2017towards}. In three passes, there exists an algorithm for $C_4$ using $\tilde{O}(\frac{m}{\epsilon^2\cdot (\# C_4)^{1/4}})$ space \cite{mcgregor2020triangle}. 

\paragraph{Cliques. ($\rho(K_r) = r/2$)} The case of counting an $r$-clique $K_r$ (for some constant $r\geq 4$) has been studied in \cite{pavan2013counting,bera2017towards}. In one pass, it is necessary to use $\Omega(\frac{m^r}{(\# K_r)^2})$ space. There exists one algorithm with $\tilde{O}(\frac{m^{r/2}}{\# K_r})$ space that uses two passes for even $r$ and four passes for odd $r$. In contrast, any multi-pass streaming algorithm requires $\Omega(\min\{\frac{m^{r/2}}{\# K_r}, \frac{m}{(\# K_r)^{1/(r-1)}}\})$ space \cite{bera2017towards}.

Finally, we mention that recently Bera and Seshadhri \cite{BerHow20a} motivated the study of streaming algorithms for subgraph counting in low \emph{degeneracy} graphs (see Definition \ref{def:degeneracy}), which is a natural class of graphs arising in practice. In addition, the class of constant degeneracy graphs includes all planar graphs, minor-closed families of graphs and preferential attachment graphs. For a graph with degeneracy at most $\lambda$, they gave a $6$-pass algorithm with space complexity $\frac{m\lambda}{\#T}\cdot \mathrm{poly}(\log n,\epsilon^{-1})$ for $(1\pm \epsilon)$-approximating the number of triangles in $G$, which breaks the worst-case lower bound for general graphs. It was conjectured that there exists a constant pass streaming algorithm for any clique with space complexity $\tilde{O}(m\lambda^{r-2}/\# K_r)$ in a graph with degeneracy at most $\lambda$ \cite{BerHow20a}.

\subsection{Our results}
Let $n$ and $m$ be the number of vertices and edges in the input graph $G$, respectively. Let $\# H$ be the number of subgraphs $H$ in $G$. Sometimes, we use $K_r$ and $T$ to denote the subgraphs $r$-clique (i.e., a clique on $r$ vertices) and triangle, respectively. Though we do not know $\# H$ in advance, we adopt the common convention from literature to parameterize our algorithms in terms of $\# H$. Since $\# H$ is unknown, $\# H$ can be replaced by a lower bound $L$ on $\# H$ to obtain corresponding guarantees for our algorithms. Alternatively, one can phrase the problem as distinguishing if the number of subgraphs $H$ is at most $L$ or at least $(1+\epsilon)L$ for an input parameter~$L$.

We first present the following algorithm.

\begin{theorem}
	\label{lem:ts-counting}
	Let $\epsilon > 0$ and let $H$ be an arbitrary subgraph of constant size. There exists a $3$-pass turnstile streaming algorithm for computing a $(1+\epsilon)$-approximation of the number of copies of $H$ in the input graph $G$ with high probability\footnote{In this paper, `with high probability' refers to `with probability at least $1-n^{-C}$, for some constant $C>0$'. } that has space complexity $\tilde{O}(m^{\rho(H)} / (\epsilon^2 \# H))$.
\end{theorem}

Note that the space complexity of our turnstile algorithm matches the insertion-only algorithm in \cite{assadi2019simple} for approximating $\# H$, while our algorithm uses only three instead of $\Omega(|V(H)|)$ passes, even in the turnstile setting (see discussion above). Furthermore, for the special case of triangles, the space complexity of our turnstile algorithm also matches the state-of-the-art of insertion-only algorithms, which is $\tilde{O}(m^{3/2} / \# T)$ \cite{mcgregor2016better,bera2017towards}, while these algorithms either use three passes together with the assumption that the algorithm is given oracle access to the degrees \cite{mcgregor2016better}, or use four passes \cite{bera2017towards}.

Our second result is a constant-pass algorithm for approximating $\# K_r$ in low degeneracy graphs, which resolves a conjecture by Bera and Seshadhri \cite{BerHow20a}. This algorithm also generalizes the algorithm in \cite{BerHow20a} that only considers $r=3$ (i.e., the triangle case).

\begin{theorem}\label{thm:clique} %
	For any $\epsilon > 0, r \geq 3$, there exists a $5r$-pass insertion-only streaming algorithm for computing an $(1+\epsilon)$-approximation to $\# K_r$ in graphs with degeneracy $\lambda$ that has space complexity
	\begin{equation*}
		\frac{m\lambda^{r-2}}{\# K_r}  \cdot \mathrm{poly}(\log n, \epsilon^{-1}, r^r)
	\end{equation*}
	and succeeds with high probability.
\end{theorem}

Both of our two main results are obtained by a generic transformation between streaming algorithms and sublinear-time algorithms in the query access model (see \cref{def:aug-gen-graph-model}). More precisely, we relate the \emph{adaptivity} of any query algorithm to the \emph{pass complexity} of a corresponding streaming algorithm that is obtained by the transformation.

\subsection{Our techniques}

The adaptivity of sublinear query algorithms, is usually classified into \emph{non-adaptive} and \emph{adaptive} algorithms. Non-adaptive algorithms must specify all queries on their input in advance, and adaptive algorithms may ask arbitrary queries during their computation (in particular, a query might depend on the previous query answers). Inspired by a notion by Cannone and Gur~\cite{CanAda18}, we define the \emph{round-adaptivity} of a sublinear-time graph query algorithm, which formalizes ``the number of levels of dependencies'' needed in an adaptive algorithm. Intuitively, each level corresponds to a set of queries that only depends on the queries in the previous levels and not on the queries on the same level (or later levels).  
Then we argue that exploiting this round-adaptivity leads to a fruitful connection between query algorithms and streaming algorithms. In particular, we show that if an algorithm is allowed to ask a batch of non-adaptive queries not just once, but for $k > 1$ rounds, this translates very naturally into a $k$-pass streaming algorithm. To illustrate our transformation, we show that known sublinear-time algorithms for sampling and counting subgraphs \cite{fichtenberger2020sampling,eden2020faster} lead to novel streaming algorithms that advance the state of the art.

\subsection{Other related work} There has been a line of works for approximately counting subgraphs in other graph stream models, including the random order model (in which the stream consists of a random permutation of the edges) \cite{mcgregor2016better,mcgregor2020triangle}, as well as in the adjacency list model (in which each edge appears twice and the edges in the stream are grouped by their endpoints) \cite{mcgregor2016better,kallaugher2019complexity}.

\section{Preliminaries}

\paragraph{Graphs.}
We consider undirected graphs. For the input graph $G=(V,E)$ of an algorithm, we use $n := \lvert V \rvert$ and $m := \lvert E \rvert$, and we denote the number of subgraphs $H$ in $G$ by $\# H$. The degree of a vertex $v$ is denoted $\dg{v}$. Our algorithms' space complexities are parameterized by the following concepts.

\begin{definition}[Fractional Edge-Cover Number]\label{def:edgecover}
    A fractional edge-cover of $H = (V_H,E_H)$ is a mapping $\psi: E_H \rightarrow [0,1]$ such that for each vertex $v\in V_H$, $\sum_{e\in E_H, v\in e} \psi(e)\geq 1$. The fractional edge-cover number $\rho(H)$ of $H$ is the minimum value of $\sum_{e\in E_H}\psi(e)$ among all fractional edge-covers $\psi$.
\end{definition}

Let $C_k$ denote the cycle of length $k$. Let $S_k$ denote a star with $k$ petals, i.e., $S_k = (\{u, v_1, \ldots, v_k\}, \cup_{i \in [k]} \{u, v_k\})$. Let $K_k$ denote a clique on $k$ vertices. It is known that $\rho(C_{2k+1})=k+1/2$, $\rho(S_k)=k$ and $\rho(K_k)=k/2$. The following result is known \cite{NPRR18,assadi2019simple}, see also \cite[Theorem 30.10]{schrijver2003combinatorial}.

\begin{lemma}
	\label{lem:subgraph-decomposition}
	For any subgraph $H$, there exist $\alpha, \beta \geq 0$ so that $H$ can be decomposed into a collection of vertex-disjoint odd cycles $\overline{C_1},\ldots,\overline{C_\alpha}$ and star graphs $\overline{S_1},\ldots,\overline{S_\beta}$ such that 
	$$\rho(H)=\sum_{i=1}^\alpha\rho(\overline{C_i})+\sum_{j=1}^\beta\rho(\overline{S_j}).$$
\end{lemma}

\begin{definition}[degeneracy]\label{def:degeneracy}
    The degeneracy of a graph is the smallest $\kappa \geq 0$ so that every subgraph has maximum vertex degree $\kappa$.
\end{definition}

\paragraph{Graph query algorithms.}
In the augmented general graph model, an algorithm gets the input size $n$ and query access to an input graph $G=(V,E)$, where $V=[n]$, and it may ask for random edges, query degrees and neighbors of vertices, and check the existence of edges in $E$. Formally, it is defined as follows.

\begin{definition}
	\label{def:aug-gen-graph-model}
    The \emph{augmented general graph model} is defined for the set of all graphs. For a graph $G=(V,E)$, where $V = [n]$, it allows four types of queries: ($f_1$) return a uniformly random edge $e \in E$; ($f_2$) given $v \in V$, return the degree of $v$; ($f_3$) given $v \in V$ and $i \in [\dg{v}]$, return the $i$\xth/ neighbor of $v$; ($f_4$) given $u,v \in V$, return whether $(u,v) \in E$. 
    
    The \emph{query complexity} of a graph query algorithm is the total number of queries it asks on its input; and its \emph{space complexity} is the maximal amount of space it uses during its execution (including space to store query answers, but excluding the space used to store the whole input).
\end{definition}

The \emph{general graph model} is the augmented general graph model without random edge queries, i.e., $f_1$.

\paragraph{Streaming algorithms.}
For our transformation, we use the following result on $\ell_0$-samplers.

\begin{lemma}[\cite{CorUni14}]
	\label{lem:l0-sampler}
	Let $c > 0$. There exists an $\ell_0$-sampler for turnstile streams on $\mathbb{Z}^n$ that requires $O(\log^4 n)$ bits of space, succeeds with probability $1-1/n^c$ and, if successful, outputs a non-zero entry $i$ with probability $1/N \pm 1/n^c$, where $N$ is the number of non-zero entries.
\end{lemma}

In pseudo code of streaming algorithms, we number the passes with respect to the current procedure, i.e., the first pass on the input in a procedure is numbered 1. Sometimes, we use parallel computation (in particular, ``\textbf{parallel for}'' loops) to enable different computations to utilize the same pass on the input. Computation that is performed during a pass on the input is placed inside a ``\textbf{pass}''-block that specifies the corresponding pass(es). For the sake of clarity, we also specify the knowledge / variables (``input'') that are available at the beginning of the pass. To make it more explicit that a streaming algorithm has queried the degrees of some vertex set $V'$, we use $\dds{V'}$ to denote a dictionary that maps every $v \in V'$ to $\dds{V'}[v] = \dg{v}$. We use $\tilde{O}(\cdot)$ to omit polylogarithmic factors and dependencies on the size of $H$.

\section{Transformation}
In this section, we present and prove a transformation that allows us to obtain constant-pass streaming algorithms from sublinear query algorithms. The number of passes depends on the level of adaptivity of the query algorithm. In particular, we consider the number of batches (\emph{rounds}) that queries can be grouped into so that a query in batch $i$ depends only on the algorithm's random coins and the answer's to queries in batches $1, \ldots, i-1$.

\begin{definition}[round-adaptive graph algorithm, cf.~\cite{CanAda18}]
    Let $\mathcal{G}$ be a set of graphs, let $\mathcal{F} = (f_1, \ldots)$ be a finite family of functions, where $f_i : \mathcal{G} \times X_i \rightarrow Y_i$ and $X_i, Y_i$ are sets, and let $k : \mathcal{G} \rightarrow \mathbb{N}$. A \emph{graph query} algorithm for a graph problem on $\mathcal{G}$ and \emph{query types} $\mathcal{F}$ can access its input $G \in \mathcal{G}$ only by asking queries $f \in \mathcal{F}$ on $G$. It is said to be \emph{$k$-round adaptive} if the following holds:
	
	The algorithm proceeds in $k(G)$ rounds. In round $\ell > 0$, it produces a sequence of queries $Q_\ell := ( \q{\ell}{i} = (\qt{\ell}{i}, \qx{\ell}{i}))_{i \in [\lvert Q_\ell \rvert]}$, where $\qt{\ell}{i} \in [\lvert \mathcal{F} \rvert]$ is a query type and $\qx{\ell}{i} \in X_{\qt{\ell}{i}}$ is the query's arguments. The sequence of queries $Q_\ell$ is based on the algorithm's own internal randomness and the answers $\mathcal{F}(Q_1), \ldots, \mathcal{F}(Q_{\ell-1})$ to the previous sequences of queries $Q_1, \ldots, Q_{\ell-1}$. In return to $Q_\ell$, the algorithms receives a sequence of query answers $\mathcal{F}(Q_\ell) = ( f_{\qt{\ell}{i}}(\q{\ell}{i}))_{i \in [\lvert Q_\ell \rvert]}$.
\end{definition}
\begin{example*}
	\normalfont Let us consider a very simple subroutine in the augmented general graph model where the goal is to find a triangle in a graph $G$. This subroutine simply does the following:
	\begin{enumerate}
		\item Sample one edge $e=(u,v)$ uniformly at random,
		\item Query the degrees of $u,v$ and find the one, say $u$, whose degree is no larger than the other,
		\item Sample a random neighbor $w$ of $u$, and
		\item Query if there exists an edge between $v$ and $w$.
	\end{enumerate}

	The above subroutine is a $4$-round adaptive graph query algorithm: In round $1$, the query set $Q_1 = ( \q{1}{1} = (1, \cdot))$ is simply one random edge and the query answer $\mathcal{F}(Q_1)$ is $(e=(u,v))$; in round $2$, the query set $Q_2 = ( \q{2}{1} = (2, u), \q{2}{2} = (3, v) )$ are the degree queries of $u,v$ and the query answer is $\mathcal{F}(Q_2)=(\dg{u}, \dg{v})$; in round $3$, the query set $Q_3 = ( \q{3}{1} = (3, (u,x) )$, where $x$ is drawn uniformly random from $[\dg{u}]$, is for a random neighbor of $u$, and the query answer is $\mathcal{F}(Q_3) = w$; in round $4$, the query set $Q_4 = ( \q{4}{1} = (4, (v,w)) )$ is if there exists an edge between $v,w$, and the query answer is $\mathcal{F}(Q_4)=(\textbf{Yes})$ if edge $(v,w)$ exists and $\mathcal{F}(Q_4)=(\textbf{No})$ otherwise.
\end{example*}

We state and prove the transformation from sublinear-time algorithms in the augmented general graph model that yields insertion-only streaming algorithms. Since the augmented general graph model subsumes the standard models for dense graphs, bounded-degree graphs and general graphs, one can directly obtain a streaming algorithm from essentially any sublinear graph query algorithm with small round-adaptivity.

\begin{theorem}
	\label{lem:reduction-io}
    Let $\mathcal{A}_Q$ be a $k$-round adaptive graph query algorithm for the augmented general graph model with query complexity $q=q(n)$ and space complexity $s=s(n)$. Then, there exists a $k$-pass algorithm $\mathcal{A}_S$ in the arbitrary-order insertion-only graph streaming model with space complexity $O(q \log n + s)$ bits so that $\mathcal{A}_S$ and $\mathcal{A}_Q$ have the same output distribution.
\end{theorem}
\begin{proof}
    Let $Q_1, \ldots, Q_k$ be the $k$ query sets that are asked by $\mathcal{A}_Q$. We define $\mathcal{A}_S$ to be the algorithm that sequentially computes $\mathcal{F}(Q_i)$, given $Q_1,\dots, Q_{i-1}$ and $\mathcal{F}(Q_1),\dots,\mathcal{F}(Q_{i-1})$, for $i\leq k$. We prove that $\mathcal{A}_S$ can compute the answers to $\mathcal{F}(Q_i)$ in pass $i$. Let $i \in [k]$, let $j \in [\lvert Q_i \rvert]$ and consider query $\q{i}{j} = (\qt{i}{j}, \qx{i}{j}) \in Q_i$. We distinguish the query type $f_{\qt{i}{j}}$ and explain how the algorithm emulates the query oracle:
    \begin{itemize}
      \item \textbf{$f_1$ (uniform edge):} A uniformly random edge can be obtained from the stream via reservoir sampling using $O(\log n)$ bits of space.
      \item \textbf{$f_2(v)$ (degree):} A counter of the degree of $v$ can be maintained while reading edges from the stream, using $O(\log n)$ bits of space.
      \item \textbf{$f_3(v,i)$ (neighbor):} The algorithm initializes a counter for $v$ that counts the number of edges read from the stream that are incidient to $v$. Once the counter reaches the value $i$, the algorithm returns $u$ from the edge $(u,v)$ it just read. This requires $O(\log n)$ bits of space.
      \item \textbf{$f_4(u,v)$ (adjacency):} The algorithm maintains a boolean variable that indicates whether the edge $(u,v)$ was read from the stream, which requires $O(\log n)$ bits of space.
    \end{itemize}
	The total space required to store all query answers is thus $O(q \log n)$. To emulate the original algorithm, one needs $O(s)$ space.
\end{proof}

To adapt sublinear graph query algorithms to \emph{turnstile} streams, we propose the following relaxed version of the augmented general graph model.

\begin{definition}
    Let $c > 0$. The \emph{relaxed augmented general graph model} is defined for the set of all graphs. For a graph $G=(V,E)$, where $V = [n]$, it allows four types of queries:
	\begin{enumerate}
		\item[($f_1$)] for every edge $e \in E$, returns $e$ with probability $1 / m \pm 1 / n^c$, or fails with probability at most $1/n^c$;
		\item[($f_2$)] given $v \in V$, returns the degree of $v$;
		\item[($f_3$)] given $v \in V$, for every $u \in \Gamma(v)$, returns $u$ with probability $1/\dg{u} \pm 1 / n^c$, or fails with probability at most $1/n^c$;
		\item[($f_4$)] given $u,v \in V$, returns whether $(u,v) \in E$.
	\end{enumerate}
	The probabilities are taken over the random coins of the respective query.
\end{definition}

This model differs in two aspects from the augmented general graph model: First, random edges that are queried via $f_1$ are not exacty uniformly random. Second, instead of asking for the $i$\xth/ neighbor of a vertex $v$ via $f_3$, one can only obtain an approximately uniformly random neighbor of $v$. Intuitively, these relaxed guarantees weaken the solution quality and the complexity of most sublinear algorithms only slightly. In particular, we prove this for subgraph counting. As a benefit of this model, we show that $k$-round adaptive graph query algorithms translate into $k$-pass \emph{turnstile} streaming algorithms.

\begin{theorem}
	\label{lem:reduction-ts}
    Let $\mathcal{A}_Q$ be a $k$-round adaptive graph query algorithm for the relaxed augmented general graph model with query complexity $q=q(n)$ and space complexity $s=s(n)$. Then, there exists a $k$-pass algorithm $\mathcal{A}_S$ in the arbitrary-order turnstile graph streaming model with space complexity $O(q \log^4 n + s)$ bits so that $\mathcal{A}_S$ and $\mathcal{A}_Q$ have the same output distribution.
\end{theorem}
\begin{proof}
    Let $Q_1, \ldots, Q_k$ be the $k$ query sets that are asked by $\mathcal{A}_Q$. We define $\mathcal{A}_S$ to be the algorithm that sequentially computes $\mathcal{F}(Q_i)$, given $Q_1,\dots, Q_{i-1}$ and $\mathcal{F}(Q_1),\dots,\mathcal{F}(Q_{i-1})$, for $i\leq k$. We prove that $\mathcal{A}_S$ can compute the answers to $\mathcal{F}(Q_i)$ in pass $i$. Let $i \in [k]$, let $j \in [\lvert Q_i \rvert]$ and consider query $\q{i}{j} = (\qt{i}{j}, \qx{i}{j}) \in Q_i$. We distinguish the query type $f_{\qt{i}{j}}$ and explain how the algorithm emulates the query oracle:
    \begin{itemize}
      \item \textbf{$f_1$ (uniform edge):} The algorithm maintains an $\ell_0$-sampler of the adjacency matrix of the graph. By \cref{lem:l0-sampler}, this requires $O((\log n^2)^4) = O(\log^4 n)$ bits of space.
      \item \textbf{$f_2(v)$ (degree):} A counter of the degree of $v$ can be maintained while reading insertions and deletions of edges that are incident to $v$ from the stream, using $O(\log n)$ bits of space.
      \item \textbf{$f_3(v)$ (random neighbor):} The algorithm maintains an $\ell_0$-sampler of the adjacency list of $v$. This requires $O(\log^4 n)$ bits of space by \cref{lem:l0-sampler}.
      \item \textbf{$f_4(u,v)$ (adjacency):} The algorithm maintains a boolean variable that indicates whether the last update of $(u,v)$ that was read from the stream was an insertion or a deletion, which requires $O(\log n)$ bits of space.
    \end{itemize}
	The total space required to store all query answers is thus $O(q \log^4 n)$. To emulate the original algorithm, one needs $O(s)$ space.
\end{proof}

\section{Subgraph counting and sampling}
\label{sec:subgraph-counting}

In this section, we analyze the round-adaptivity of the sublinear algorithm for sampling uniformly random copies of a given subgraph $H$ by Fichtenberger, Gao and Peng \cite{fichtenberger2020sampling}, which can also be easily adapted to obtain a subgraph counting algorithm,  in the augmented general graph model. We show that this algorithm is $3$-round adaptive. Therefore, it yields a $3$-pass streaming algorithm for sampling and counting subgraphs via \cref{lem:reduction-io}. 

\subsection{A sublinear-time algorithm for counting arbitrary subgraphs}

We make use of a subroutine from \cite{fichtenberger2020sampling} for sampling a copy of subgraph $H$ in $G$, and we refer to this subroutine as the \emph{FGP algorithm} in the following. To describe the FGP algorithm, we state the relevant definitions.

\begin{definition}[vertex order]
	Let $G=(V,E)$ be a graph, let $u, v \in V$. We define $u \prec_G v$ if and only if $\dg[G]{u} < \dg[G]{v}$, or $\dg[G]{u} = \dg[G]{v}$ and $id(u) < id(v)$.
\end{definition}

\begin{definition}[canonical cycle]
	\label{def:can_cycle}
	Let $G=(V,E)$ be a graph and let $E' \subseteq E$. A sequence of vertices $(u_1, \ldots, u_k)$ is a canonical $k$-cycle in $(E', \prec_G)$ if, for all $i \in [k]$, $(u_i, u_{i + 1 \mod{k+1}}) \in E'$ and, for $2 \leq i \leq k$, $u_1 \prec u_{i}$ and $u_k \prec u_2$.
\end{definition}

\begin{definition}[canonical star]
	\label{def:can_star}
	Let $G=(V,E)$ be a graph and let $E' \subseteq E$. A sequence of vertices $(u_0 u_1, \ldots, u_k)$ is a canonical $k$-star in $(E', \prec_G)$ if, for all $i \geq 1$, $(u_0, u_i) \in E'$, and, for $1 \leq i < k$, $u_i \prec u_{i+1}$.
\end{definition}

\paragraph{High-level idea of the FGP algorithm} The idea of the FGP algorithm is to first compute a decomposition of the subgraph $H$ into odd cycles and stars according to \cref{lem:subgraph-decomposition}, and design subroutines to sample each canonical cycle of length $2k+1$ with probability $1/(2m)^{k+1/2}$, and each canonical $k$-petal star with probability $1/(2m)^{k}$. Together with the relation between the fractional edge cover $\rho(H)$ of a subgraph $H$ and its decomposition into odd length cycles and stars, the FGP algorithm then uses these samples to output a subgraph such that for any copy of $H$, it is output with probability $1/(2m)^{\rho(H)}$. %

More precisely, the algorithms tries to sample a copy of $H$ by sampling canonical cycles and stars according to \cref{def:can_cycle,def:can_star}. For a canonical $k$-star $(u_0, \ldots, u_k)$, it simply samples $k$ random edges $(v_1, w_1), \ldots, (v_k,w_k)$ and checks if the sampled edges form indeed a $k$-star subgraph, $v_1 = \ldots = v_k$ and $w_i \prec w_{i+1}$ for all $i \in [k-1]$. For a canonical odd cycle $(u_1, \ldots, u_{2k+1})$, it tries to sample every second edge, i.e., $(u_1, u_2), (u_3, u_4), \ldots, (u_{2k-1}, u_{2k})$. Note that only $u_{2k+1}$ is still unknown to the algorithm. Now, the algorithm proceeds based on the following case distinction: either, all the vertices in the cycle have degree greater than $\sqrt{2m}$. Then, one endpoint of a uniformly random edge is $u_{2k+1}$ with probability $\dg{u} / 2m \approx 1/\sqrt{2m}$. Otherwise, it samples the $i$\xth/ neighbor of $u_1$ (if it exists), where $i$ is drawn uniformly at random from $[\sqrt{2m}]$. Since $u_1$ has degree less than $\sqrt{2m}$, $u_{2k+1}$ is sampled with probability $1/\sqrt{2m}$. Then, the algorithm checks whether the sampled edges form a cycle of length $2k+1$, and whether it is canonical.

For the sake of completeness, we provide pseudo code of the FGP algorithm (i.e., \cref{subgraph-sample} \textsc{SampleSubgraph}) in \cref{sec:sublineartimesubgraph}. Its performance guarantee is given in the following lemma. 
\begin{lemma}[\cite{fichtenberger2020sampling}, Lemma 8]
	\label{lem:query-sampler}
	Let $H$ be an arbitrary subgraph of constant size. The FGP algorithm takes an input graph $G=(V,E)$ and the number of edges $m$ in $G$, uses $O(1)$ queries in expectation and guarantees the following: For a fixed copy of $H$ in $G$, the probability that $H$ is returned is $1 / (2m)^{\rho(H)}$.
\end{lemma}

\subsection{Insertion-only streaming algorithm}
\label{sec:subgraph-counting-io}

For the sake of presentation, we start with an insertion-only algorithm. Since \cref{lem:reduction-io} transforms round-adaptive query algorithms in the (standard) augmented general graph model to insertion-only streaming algorithms, our only objective in this section is to prove the round-adaptivity of the FPG algorithm.

\begin{lemma}
	\label{lem:stream-sampler}
	Let $H$ be an arbitrary subgraph of constant size. There exists a $3$-pass insertion-only streaming algorithm (i.e., \cref{alg:stream-subgraph}) that has space complexity $O(\log n)$ and returns a copy of $H$ or nothing. For any fixed copy of $H$ in the input graph, it is returned with probability $1 / (2m)^{\rho(H)}$.
\end{lemma}
\begin{proof}
We prove that the FGP-algorithm (see \cref{subgraph-sample} \textsc{SampleSubgraph} in \cref{sec:sublineartimesubgraph}) is $3$-round adaptive by devising a partition of its queries into $3$ rounds.

First, the FGP-algorithm computes a decomposition of $H$ into $\alpha \geq 0$ odd cycles with lengths $c_1, \ldots, c_\alpha$ and $\beta \geq 0$ stars with $s_1, \ldots, s_\beta$ petals according to \cref{lem:subgraph-decomposition} without making any queries. In the first round, the FGP-algorithm samples a set of edges that will be used to form potential cycles and stars; in the second round, the algorithm samples a random neighbor of some vertex in each of the potential cycles; in the third round, the algorithm performs vertex pair queries to check if these potential cycles and stars are indeed cycles and stars, respectively. Then we run a postprocessing on the collected subgraph and output a copy $H$ if it is found. %
By the proof of \cref{lem:query-sampler}, any copy of $H$ is output with probability $1 / (2m)^{\rho(H)}$.
	
	For the sake of presentation, we state the algorithm as a streaming algorithm in \cref{alg:stream-subgraph}. The space complexity of this algorithm then follows from \cref{lem:query-sampler,lem:reduction-io}.

	We argue that \cref{alg:stream-subgraph} is indeed a 3-pass algorithm. To emulate the queries, we invoke \cref{lem:reduction-io}. In the first pass, the algorithm only needs to know the decomposition of $H$ to sample sufficiently many edges, and it computes the number of edges $m$. Prior to the second pass, it needs to know $m$ to sample $i$ from $\{1, \ldots, \sqrt{2m}\}$. The third pass checks only the existence of edges between the known vertices in $V' = (C'_i)_{i \in [\alpha]} \cup (S_i)_{i \in [\beta]}$ and computes their degrees. Afterwards, the algorithm has obtained the subgraph induced by all vertices and the degrees of these vertices in $G$.
	
	We note that the information collected by \cref{alg:stream-subgraph} is enough to check whether the sampled cycles and stars are canonical and to check whether $G[V']$ spans or induces a copy of $H$, which are the only checks performed by \cref{line:checkstart} -- \cref{line:checkstop}.
\end{proof}

Since there are $\# H$ copies of $H$ in $G$, the probability that \cref{alg:stream-subgraph} returns a copy of $H$ is $\# H / (2m)^{\rho(H)}$ by \cref{lem:stream-sampler}. Then the subgraph counting algorithm can be obtained by viewing the above subgraph sampler as a biased coin. That is, let $p=\frac{\# H}{(2m)^{\rho(H)}}$ be the probability of a coin getting a \textsc{Heads} on a flip (which corresponds to a copy of $H$ returned by \cref{alg:stream-subgraph}). By a standard Chernoff bound argument, one can obtain a multiplicative approximation of $p$ by flipping it sufficiently many times and counting how often it turns up heads. Formally, we have the following theorem.

\begin{theorem}
	\label{lem:io-counting}
	Let $\epsilon > 0$ and let $H$ be an arbitrary subgraph of constant size. There exists a $3$-pass insertion-only streaming algorithm for computing a $(1+\epsilon)$-approximation of the number of copies of $H$ in the input graph $G$ with high probability that has space complexity $\tilde{O}(m^{\rho(H)} / (\epsilon^2 L))$, where $L$ is a lower bound on $\# H$.
\end{theorem}
\begin{proof}
	We run $k = 30 (2m)^{\rho(H)} \ln n \, / \, (\epsilon^2 L)$ copies of \cref{alg:stream-subgraph} in parallel. The probability that a single instance returns a subgraph is $\# H / (2m)^{\rho(H)}$ by \cref{lem:stream-sampler}. Let $x$ denote the fraction of invocations that returned a subgraph. Since all invocation are independent of each other, using Chernoff bound it follows that $\lvert \# H - (2m)^{\rho(H)} x \rvert \leq \epsilon \cdot \# H$ with probability at least $1-\frac{1}{n^{\Omega(1)}}$. Since each instance of the algorithm requires $O(\log n)$ space by \cref{lem:stream-sampler}, the total space bound is $O(k \log n) = O(m^{\rho(H)} \log^2(n) \, / \, (\epsilon^2 L))$.
\end{proof}

\newcommand{\subgraphAlgorithms}{
\begin{algorithm}
    \caption{Sampling a subgraph candidate from a stream}
    \tagged{io_version}{\label{alg:stream-subgraph}}
	\tagged{ts_version}{\label{alg:turnstile-subgraph}}
    \begin{algorithmic}[1]
    \Procedure{\textsc{StreamSubg}}{$C = (c_1, \ldots, c_\alpha), S = (s_1, \ldots, s_\beta)$}
        \Pass{1}{$C,S$}
			\ForAll{$c_i \in C$}
				\State sample $\lceil c_i / 2 \rceil + 1$ edges $C_i =$
				\StatexIndent[4] $\{u_{i,0}, v_{i,0}, \ldots, u_{i,\lceil c_i / 2 \rceil}, v_{i,\lceil c_i / 2 \rceil}\}$ \Comment{$f_1$}
			\EndFor
			\ForAll{$s_i \in S$}
				\State sample $k$ edges $S_i = \{x_{i,1}, y_{i,1}, \ldots, x_{i,k}, y_{i,k}\}$ \Comment{$f_1$}
			\EndFor
			\State count the number of edges $m$
		\EndPass
		\Pass{2}{$C,S,(C_i)_{i \in [\alpha]}, (S_i)_{i \in [\beta]}, m$}
			\ForAll{$c_i \in C$}
				\tagged{io_version}{
					\State sample $j \in [\sqrt{2m}]$ uniformly at random \label{lin:f3-a}
					\State $w_i \gets$ $j$\xth/ neighbor of $u_{i,1}$ \Comment{$f_3$} \label{lin:f3-b}
				}
				\tagged{ts_version}{
					\State $w_i \gets$ random neighbor of $u_{i,1}$ \Comment{$f_3$}
				}
				\State $C'_i = C_i \cup \{w_i\}$
			\EndFor
		\EndPass
		\Pass{3}{$C,S,(C'_i)_{i \in [\alpha]}, (S_i)_{i \in [\beta]}, m$}
			\ForAll{$z,z' \in (C'_i)_{i \in [\alpha]} \cup (S_i)_{i \in [\beta]}$}
				\State check whether $(z,z') \in E$ \Comment{$f_4$}
				\State count degree of $z$ \Comment{$f_2$}
			\EndFor
		\EndPass
		\State Let $V' := (C'_i)_{i \in [\alpha]} \cup (S_i)_{i \in [\beta]},$
		\StatexIndent[2] $E' := E \cap (V' \times V'), V'' := \emptyset$
\State
\State \Comment{Postprocessing}
        \ForAll{$c_i \in C'_i$}
			\If{$\dds{V'}[u_{i,1}] \leq \sqrt{2m}$} \tagged{io_version}{\label{lin:f3-repair}}
				\tagged{ts_version}{
					\State sample $t \in [\sqrt{2m}]$ uniformly at random
					\State check if $t \leq \dg{u_{i,1}}$
				}
				\State check if $(w_{i}, u_{i,1}, v_{i,1}, \ldots, u_{i,\lceil c_i/2 \rceil}, v_{i,\lceil c_i/2 \rceil})$
				\StatexIndent[4] is a canonical $c_i$-cycle in $(E', \prec_G)$
				\State $V'' = V'' \cup \{ w_{i}, u_{i,1}, \ldots, u_{i,c_i}, v_{i,1}, \ldots, v_{i,\alpha} \}$
			\Else
				\State sample $t \in [0,1]$ uniformly at random
				\State check if $t \leq \sqrt{2m} / \dg{u_{i,0}}$
				\State check if $(u_{i,0}, u_{i,1}, v_{i,1}, \ldots, u_{i,\lceil c_i/2 \rceil}, v_{i,\lceil c_i/2 \rceil})$
				\StatexIndent[4] is a canonical $c_i$-cycle in $(E', \prec_G)$
				\State $V'' = V'' \cup \{ u_{i,0}, u_{i,1}, \ldots, u_{i,c_i}, v_{i,1}, \ldots, v_{i,\alpha} \}$
			\EndIf
		\EndFor
		\ForAll{$s_i \in S_i$}
			\State\label{line:checkstart} check if $x_{i,1} = \ldots = x_{i,s_i}$ and $(x_{i,1}, y_{i,1}, \ldots, y_{i,s_i})$
			\StatexIndent[3] is a canonical $s_i$-star in $(E'', \prec_G)$
			\State $V'' = V'' \cup \{ x_{i,1}, y_{i,1}, \ldots, y_{i,s_i} \}$
		\EndFor
		\State\label{line:checkstop} check if $G[V'']$ contains a copy $H_G$ of $H$ as subgraph
		\State \Return $H_G$ if all checks were successful
	\EndProcedure
	\end{algorithmic}
\end{algorithm}
}
\usetag{io_version}
\subgraphAlgorithms

\subsection{Turnstile streaming algorithm}
\label{sec:subgraph-counting-ts}

In this section, we adapt the analysis of the insertion-only algorithm from \cref{sec:subgraph-counting-io} to the \emph{relaxed} augmented general graph model. In particular, we show that essentially the same algorithm yields \emph{approximately} uniformly randomly sampled subgraphs, which in turn is still sufficient to approximately count subgraphs. 

\begin{lemma}[{restate=[name=]turnstileSampler}]
	\label{lem:turnstile-sampler}
	Let $\epsilon > 0$, and let $H$ be an arbitrary subgraph of constant size. There exists an algorithm for the relaxed augmented general graph model that has space complexity $O(\log^4 n)$ and returns a copy of $H$ or nothing. For any fixed copy of $H$ in the input graph, it is returned with probability $(1 \pm \epsilon)/(2m)^{\rho(H)}$.
\end{lemma}
\newcommand{\turnstileSamplerProof}{
\begin{proof}
	Let $C = (c_1, \ldots, c_\alpha), S = (s_1, \ldots, s_\beta)$ be a decomposition of $H$ into odd-length cycles and stars that satisfies the guarantees of \cref{lem:subgraph-decomposition} and that is passed to \cref{alg:stream-subgraph}. We slightly modify \cref{alg:stream-subgraph} as follows. In \cref{alg:stream-subgraph}, we replace \cref{lin:f3-a,lin:f3-b} by setting $w_i$ to the answer of a query $f_3(u_{i,1})$, i.e., an approximately uniformly random neighbor of $u_{i,1}$. %
	After \cref{lin:f3-repair}, we sample a uniformly random number $t$ from $[\sqrt{2m}]$ and check whether it is at most $\dg{u_{i,1}}$. Since no additional queries are asked, the space complexity follows from \cref{lem:query-sampler,lem:reduction-ts}.
	\tagged{mainpart}{Pseudo code and the proof of correctness are provided in \cref{sec:app-turnstile}.}%
	\tagged{appendix}{
	The modified algorithm is provied in \cref{alg:turnstile-subgraph}.

	Let $i \in [\beta]$ and set $k := s_i$. Let $(u_0, \ldots, u_k)$ be a \emph{canonical} $k$-star in $G$, and let $p$ be the probability that it is sampled by queries of type $f_1$. Then, we have $(1/(2m) - 1/n^c)^k \leq p \leq (1/(2m) + 1/n^c)^k$.

	Let $i \in [\alpha]$ and set $k := (c_i-1)/2$. Let $(u_1, \ldots, u_{2k+1})$ be a \emph{canonical} odd-length cycle in $G$. Let $p$ be the probability that $((u_1,u_2), (u_3,u_4), \ldots, (u_{2k-1}, u_{2k}) )$ are sampled via queries of type~$f_1$. Then, we have $(1/(2m) - 1/n^c)^{k} \leq p \leq (1/(2m) + 1/n^c)^{k}$. Let $q$ be the probability that $u' := u_{2k+1}$ is sampled. We distinguish two cases: $\dg{u_1} \leq \sqrt{2m}$ and $\dg{u_1} > \sqrt{2m}$. If $\dg{u_1} \leq \sqrt{2m}$, then $u' = w_i$ is sampled according to the modification described above via a query of type $f_3$ and 
	\begin{equation*}
		\frac{\dg{u_1}}{\sqrt{2m}} \cdot \left( \frac{1}{\dg{u_1}} - \frac{1}{n^c} \right)
		\leq q
		\leq \frac{\dg{u_1}}{\sqrt{2m}} \cdot \left( \frac{1}{\dg{u_1}} + \frac{1}{n^c} \right) .
	\end{equation*}
	Now, consider the case $\dg{u'} > \sqrt{2m}$. Then, $u' = u_{i,0}$ is sampled via a query of type $f_1$. Note that sampling a vertex \emph{exactly} proportional to its degree is equivalent to sampling an edge \emph{exactly} uniformly at random and choosing one of its endpoints by flipping a fair coin. Since $f_1$ returns exactly one edge, for any $e,e' \in E$, $e \neq e'$, the events of sampling $e$ and sampling $e'$ are disjoint. For the corresponding query of type $f_1$ in the \emph{relaxed} augmented general graph model it follows that
	\begin{equation*}
		\frac{\sqrt{2m}}{\dg{u'}} \cdot \left( \frac{\dg{u'}}{2m} - \frac{\dg{u'}}{n^c} \right)
		\leq q
		\leq \frac{\sqrt{2m}}{\dg{u'}} \cdot \left( \frac{\dg{u'}}{2m} + \frac{\dg{u'}}{n^c} \right) .
	\end{equation*}

	Fix a copy of $H$ and let $p$ be the probability that $H$ is sampled. It follows from the discussion above that
	\begin{equation*}
		\frac{1}{(2m)^{\rho(H)}} - \frac{2^{\lvert H \rvert}n}{n^c}
		\leq p
		\leq \frac{1}{(2m)^{\rho(H)}} + \frac{2^{\lvert H \rvert}n}{n^c}
	\end{equation*}
	Choosing $c = 5 \lvert H \rvert \geq \log(2^{\lvert H \rvert} n (2m)^{\rho(H)} \, / \, \epsilon) \, / \, \log(n)$ concludes the proof.
	}
\end{proof}
}
\turnstileSamplerProof

We can obtain the subgraph counting algorithm in the turnstile streaming model similarly as we prove  \cref{lem:io-counting} by considering the sampler as a biased coin and estimating its heads probability. Now we are ready to prove Theorem \ref{lem:ts-counting}.

\begin{proof}[Proof of Theorem \ref{lem:ts-counting}]
Consider the algorithm from \cref{lem:turnstile-sampler} (i.e., \cref{alg:turnstile-subgraph}) with $\varepsilon=\frac{\epsilon}{3}$. Let $p$ be the probability that some copy of subgraph $H$ is returned. By \cref{lem:turnstile-sampler}, $p=\# H (1 \pm \frac{\epsilon}{3}) / (2m)^{\rho(H)}$. We can think of the above algorithm as tossing a coin with bias $p$. By the Chernoff bound, with high probability, the bias $p$ can be estimated up to a multiplicative factor $(1\pm \frac{\epsilon}{3})$ in $O(\log n/\varepsilon^2p)=\tilde{O}(m^{\rho(H)} / (\epsilon^2 \# H))$ tosses, which can be implemented by running in parallel the same number of copies of \cref{alg:turnstile-subgraph}. Given the estimate $p$, $\#H$ can be approximated within a multiplicative factor $(1\pm \epsilon)$. Since each instance of the algorithm requires $O(\log^4 n)$ space by \cref{lem:turnstile-sampler}, the total space bound is $\tilde{O}(m^{\rho(H)} / (\epsilon^2 \# H))\cdot O( \log^4 n) = \tilde{O}(m^{\rho(H)} / (\epsilon^2 \# H))$. This finishes the proof of the theorem.
\end{proof}

\let\saveFloatBarrier\FloatBarrier%
\let\FloatBarrier\relax%
\section{Low-degeneracy clique counting}
\let\FloatBarrier\saveFloatBarrier%
\subsection{A sublinear-time algorithm for counting cliques in a graph with low degeneracy}
Now we describe the sublinear-time algorithm in the general graph model (i.e., the augmented general graph model without edge sampling queries) for approximating $\# K_r$ of a graph with degeneracy at most $\lambda$, for any $r\ge 3$ and $\lambda >0$. The algorithm is given by Eden, Ron and Seshadhri \cite{eden2020faster}. In the following, we call the algorithm in \cite{eden2020faster} the \emph{ERS algorithm}.

\paragraph{High-level description of the ERS algorithm.} 
For the sake of a concise presentation, we will assume that the algorithm is given $\#K_r$. 
To obtain an estimate when only a lower bound $L$ on $\# K_r$ is available, it is straightforward to use (parallel) geometric search for values greater than $L$ (see \cref{thm:fromerstheorem} in \cref{app:algo_degeneracy}).
The ERS algorithm makes use of a notion of \emph{ordered} cliques. More precisely, for any $t\in \{2,\dots,r\}$, an \emph{ordered $t$-clique} $\vec{T}=(v_1,\cdots,v_t)$ is a tuple of $t$ vertices such that $T=\{v_1,\cdots,v_t\}$ forms a $t$-clique (and $T$ is called an \emph{unordered $t$-clique}).

Let $s_1,\dots, s_r$ be some parameters. The ERS algorithm iteratively does the following: in iteration $1$, it samples a set $\RR_1$ of $s_1$ ordered vertices; in iteration $2$, it samples a set $\RR_2$ of $s_2$ ordered edges (incident to the sampled vertices); and then in iteration $t>2$, it samples a set $\RR_t$ of $s_t$ ordered $t$-cliques, based on the set of $(t-1)$-cliques from the previous iteration. Concretely, the sample set $\RR_{t}$ is obtained by repeating the following $s_t$ times:
\begin{enumerate}
    \item sample an ordered clique $\vec{T}$ from $\RR_{t-1}$ with probability proportional to $\frac{\dg{\vec{T}}}{\dg{\RR_{t-1}}}$, where $\dg{\vec{T}}$ is the degree of the \emph{minimum-degree} vertex in $\vec{T}$ and for a set of order cliques $\RR$, $\dg{\RR}:=\sum_{T\in \RR}\dg{\vec{T}}$;
	\item select a uniformly random neighbor $w$ of the least degree vertex in $\vec{T}$; and
    \item check if $\vec{T}$ and $w$ forms a $t$-clique, and if so, add it to $\mathcal{R}_t$.
\end{enumerate}

Once we have these sampled sets $\RR_1,\cdots,\RR_r$, we can use their \emph{weights} to approximate $\#K_r$. Roughly speaking, for each ordered $t$-clique $\vec{T}$, its \emph{weight} $\ww(T)$ is a number that is close to the number of $r$-cliques that are assigned to $\vec{T}$ according to an \emph{assignment rule} specified below. Throughout the process, the ERS algorithm carefully chooses the parameters so that $\ww(\RR_1)$ is close to $\frac{\# K_r}{n}\cdot s_1$, and for any $t\geq 1$, $\ww(\RR_{t+1})$ is close to $\frac{\ww(\RR_t)}{\dg{\RR_t}}\cdot s_{t+1}$. Thus, it suffices to compute $\ww(\RR_r)$ for estimating $\#K_r$, as for any $t\geq 2$, $\ww(\RR_t)$ is close to $\frac{\# K_r}{n}\cdot\frac{s_1\cdot \dots\cdot s_t}{\dg{\RR_1}\cdot\dots\cdot \dg{\RR_{t-1}}}$.%

To guarantee the above, the ERS algorithm defines an \emph{assignment rule} to assign each $r$-clique $C$ to an ordered $t$-clique $\vec{T}$, for any $t\leq r$. To check if an ordered $r$-clique $\emph{C}$ is assigned the corresponding unordered clique, the algorithm 
\begin{enumerate}
    \item considers, for every $t\in  \{2,\dots,r\}$, all the prefixes $\vec{C}_{\leq t}$, where a prefix $\vec{C}_{\leq t}$ of $\vec{C}$ is the ordered $t$-clique whose vertices are the first $t$ vertices in $\vec{C}$. For each prefix $\vec{C}_{\leq t}$, it invokes a subroutine \textsc{IsActive} to check if it is \emph{active}, %
    which in turn iteratively samples sets $\RR_{t+1},\cdots,\RR_{r}$, starting from $\RR_t=\{\vec{C}_{\leq t}\}$ and decides the activeness by the statistics of these sample sets;
 	\item if $\vec{C}$ is the lexicographically smallest active ordered  $r$-clique among all active ordered $r$-cliques induced by $C$, then it returns $1$ (indicating that $\emph{C}$ is assigned the corresponding unordered clique); otherwise, it returns $0$. 
\end{enumerate}

For the sake of completeness, we provide pseudo code of this algorithm (i.e., \cref{clique-estimation} \textsc{CountClique}) in \cref{app:algo_degeneracy}. 

\paragraph{Simplifying the ERS algorithm in the augmented graph model.}  The above ERS algorithm in \cite{eden2020faster} was described in the general graph model, in which the algorithm can not perform edge sampling queries. The authors of \cite{eden2020faster} need to carefully select $s_1$ so that it can handle different cases of $\# K_r$ and remedy the defect of not being able to sample uniform edges. The choice of $s_1$ causes an additional term $\min\{\frac{n\lambda^{r-1}}{\#K_r}, \frac{n}{(\# K_r)^{1/r}}\}$ in the query complexity $\min \left\{ \frac{n\lambda^{r-1}}{\# K_r}, \frac{n}{(\# K_r)^{1/r}} + \frac{m\lambda^{r-2}}{\# K_r} \right\} \cdot \mathrm{poly}(\log n, 1/\epsilon, r^r)$ of their algorithm (see Theorem 1.1 in \cite{eden2020faster} and also \cref{thm:fromerstheorem} in \cref{app:algo_degeneracy}).

We note that this algorithm can be simplified in the augmented general graph model, which can be further transformed to the streaming setting by \cref{lem:reduction-ts}. More precisely, we note that in the augmented graph model, one can directly start with sampling a set $\RR_2$ of $s_2$ edges independently and uniformly at random, and then iteratively sample a set of $t$-cliques $\RR_t$ based on $\RR_{t-1}$, for any $3\leq t\leq r$. That is, there is no need to sample a set $\RR_1$ of vertices (and we simply set $\RR_1:=E(G)$ and $\dg{\RR_1}:=m$ at the beginning of the algorithm). Then we choose parameters to ensure that $\ww(\RR_2)$ is close $\frac{\#K_r}{m}\cdot s_2$, where $\ww(\RR_2)$ is the total weight of ordered edges in $\RR_2$ and the weight of an ordered edge is the number of $r$-cliques assigned to it. Then one can still guarantee that for any $t\geq 3$, $\ww(\RR_t)$ is close to $\frac{\# K_r}{m}\cdot\frac{s_2\cdot \dots\cdot s_t}{\dg{\RR_2}\cdot\dots\cdot \dg{\RR_{t-1}}}$, by setting the parameters $s_2,\cdots,s_r, \vec{\tau}$ similarly as in \cite{eden2020faster} (while we start with a slightly different choice $s_2$ due to the uniform edge sampling). By the analysis in \cite{eden2020faster}, we have the following lemma regarding the performance guarantee of the ERS algorithm in the augmented graph model. 
\begin{lemma}[\cite{eden2020faster}]
	\label{lem:ers-algo}
	Given query access to a graph with degeneracy $\lambda$ in the augmented general graph model, the ERS algorithm has expected running time and query complexity
	\begin{equation*}
	\frac{m\lambda^{r-2}}{\# K_r}  \cdot \mathrm{poly}(\log n, 1/\epsilon, r^r)
	\end{equation*}
	and outputs a $(1+\epsilon)$-approximation to $\# K_r$ with high probability. Additionally, the maximum query complexity is $O(m+n)$.
\end{lemma}

\subsection{The round-adaptivity of the ERS algorithm}
\label{sec:streaming-degeneracy}
Now we show that the ERS algorithm in the augmented general graph model is an $O(r)$-round algorithm. We describe our streaming version of the ERS algorithm, based on the discussion before. The first observation is that the algorithm consists of two sequentially aligned blocks: sampling the sets $(\RR_i)_{i \in \{2,\dots,r\}}$ and checking the assignments. In both blocks, $(i+1)$-cliques are iteratively constructed from $i$-cliques. These constructions are also inherently sequential. However, constructing multiple $i$-cliques and some other computations can be done in parallel. Details follow below.

We state the pseudo code of the $5r$-pass streaming version of the ERS in \cref{alg:stream-approx,alg:stream-sampleset}. To obtain the final result, we use probability amplification and return the median from running sufficiently many and accordingly parameterized instances of \cref{alg:stream-approx}, which is described in \cref{stream-clique-estimation}.

\paragraph{Construction of $(\RR_t)_{t \in \{2,\dots,r\}}$}
The algorithm \textsc{StreamApproxClique} (\cref{alg:stream-approx}) constructs sets of ordered $t$-cliques $\RR_t$ iteratively, for $t = 2, \ldots, r$. Given a set $\RR_t$, it calls a $2$-pass procedure \textsc{StreamSet} to construct $\RR_{t+1}$. After the construction of $\RR_r$, it checks how many of the sampled ordered cliques in $\RR_t$ are cliques that are assigned their respective unordered clique using a $2r$-pass procedure \textsc{StrIsAssigned} (see \cref{sec:app-missing}), and it outputs this number scaled accordingly as its estimate of $\# K_r$.

\paragraph{Sampling $\RR_{t+1}$.}
Given $\RR_t$, the procedure \textsc{StreamSet} (\cref{alg:stream-sampleset}) samples up to $s_{t+1}$ many ordered $(t+1)$-cliques to include into $\RR_{t+1}$. To sample \emph{one} ordered $(t+1)$-clique, the algorithm samples an ordered $t$-clique $\vec{T}$ from $\RR_{t}$ proportionally to $\dg{\vec{T}}$. We note that the algorithm can maintain a data structure $\dds{\RR_{t}}$ for every $t \in \{2,\dots,r\}$ so that this sampling can be done offline without a pass on the input. To select a uniformly random neighbor of the smallest-degree vertex $u$ of $\vec{T}$, the algorithm samples $i \in [\dg{u}]$ and queries the $i$\xth/ neighbor of $u$ in a single pass. In another pass, the algorithm checks whether $(\vec{T}, w)$ is a $(t+1)$-clique and adds it to $\RR_{t+1}$ if this is the case. Sampling $s_{t+1}$ ordered $(t+1)$-cliques like this can be parallelized.

\paragraph{Checking the assignments of $\RR_r$.}
Given $\vec{C} \in \RR_r$, for an ordered $r$-clique $\vec{C}'$ that is isomorphic to $\vec{C}$ and a prefix $\vec{C}'_{\leq t}$, where $t \in  \{2,\dots,r\}$, \textsc{StrIsAssigned} (\cref{alg:stream-assigned} in \cref{sec:app-missing}) calls \textsc{StrAct} (\cref{alg:stream-active} in \cref{sec:app-missing}). The latter is similar to a ``warm-start'' of (multiple instances of) \textsc{StreamApprox} with $\RR_t = \{ \vec{C}'_{\leq t} \}$: it uses $2(r-t)$ passes to iteratively construct sets $\RR_{t+1}, \ldots, \RR_r$ via \textsc{StreamSet}. Then, it returns whether sufficiently many instances of \textsc{StreamApprox} satisfy a threshold condition that can be computed offline.

Since the parameters of a call to \textsc{StrAct} do not depend on calls for other (shorter) prefixes of $\vec{C}'$, all calls corresponding to different prefixes can be parallelized. Note that this differs from the iterative construction of $(\RR_t)_{t \in  \{2,\dots,r\}}$ in \textsc{StreamApprox}. In addition, the calls corresponding to different ordered $r$-cliques can be parallelized. Once the algorithm has determined, for all ordered $r$-cliques and their prefixes, whether they are active, the fully active ordered $r$-cliques can be compared lexicographically against the sampled $r$-clique $\vec{C}$. If $\vec{C}$ is active and lexicographically smallest, the algorithm accepts, otherwise it rejects. It follows that due to parallel computation, we require as many passes as the most costly call to \textsc{StrAct}, which is at most $2r$ passes.

\begin{theorem}
	The ERS algorithm in the augmented graph model can be implemented as a $5r$-round adaptive algorithm.
\end{theorem}

Applying \cref{lem:reduction-io} proves \cref{thm:clique}.

\begin{algorithm}
	\caption{Approximately counting the number of $K_r$}
	\label{stream-clique-estimation}
	\begin{algorithmic}[1]  
		\Procedure{StreamCountClique}{$n, r,\lambda,\epsilon$}
		\State $\gamma \gets \epsilon/(8r\cdot r!)$, $\beta\gets 1/(6r)$
		\State $\tau_r\gets 1$; for each $t\in [2,r-1]$, set $\tau_t\gets\frac{r^{4r}}{\beta^r\cdot \gamma^2}\cdot \lambda^{r-t}$; \StatexIndent[1] $\vec{\tau}\gets\{\tau_2,\dots,\tau_r\}$
		\ForAll{$j=1,\ldots, q=\Theta(\log (n))$}
		\State Invoke \textsc{StreamApproxCliques}($n,r,\lambda,\epsilon,  \vec{\tau}$).%
		\State Let $\chi_j$ be the returned value. 
		\EndFor
		\State Let $\hat{n}_r$ be the the median value of $\chi_1,\cdots, \chi_q$ \State \Return $ \hat{n}_r$.
		\EndProcedure
	\end{algorithmic}
\end{algorithm}

\begin{algorithm}
	\caption{The basic subroutine}
	\label{alg:stream-approx}
	\begin{algorithmic}[1]  
		\Procedure{StreamApproxClique}{$n,r,\lambda,\epsilon, \vec{\tau}$}
		\State set $\tilde{\ww}_1 \gets(1-\epsilon/2)\# K_r$, $\beta\gets 1/(18r)$, $\gamma\gets \epsilon/(2r)$
		\State virtually set $\RR_1\gets E$
		\Pass{1}{}
			\State count number of edges $m$ and set $\dg{\RR_1}\gets m$
		\EndPass
		\Pass{2}{m}
			\State sample $s_2 \gets \lceil \frac{m\cdot\tau_2}{\tilde{\ww}_0}\cdot \frac{3\ln(2/\beta)}{\gamma^2}\rceil$ edges u.a.r
			\StatexIndent[3] and let $\RR_2$ be the chosen multiset
			\Comment{$f_1$}		
		\EndPass
		\Pass{3}{$\RR_2$}
			\State construct $\dds{\RR_2}$ \Comment{$f_2$}
		\EndPass
		\ForAll{$t \in \{2,\dots,r-1\}$}
			\State $\dg{\RR_t} \gets \sum_{\vec{T} \in \RR_t} \dg{\vec{T}} = \sum_{\vec{T} \in \RR_t} \min_{v \in \vec{T}} \dds{\RR_t}[v]$
			\State $\tilde{\ww}_t \gets (1-\gamma)\frac{\tilde{\ww}_{t-1}}{\dg{\RR_{t-1}}}\cdot s_t$
			\StatexIndent[3] and $s_{t+1}\gets \lceil \frac{\dg{\RR_t}\tau_{t+1}}{\tilde{\ww}_t}\cdot \frac{3\ln(2/\beta)}{\gamma^2}\rceil$
			\State \textbf{if} $s_{t+1}>\frac{4m\lambda^{t-1}\cdot \tau_{t+1}}{\#K_r}\cdot \frac{(r!)^2\cdot 3\ln (2/\beta)}{\beta^t\cdot \gamma^2}$, \textbf{then} \textbf{abort}
			\Passes{2t}{2t+1}{$t,\RR_t,\dds{\RR_t},s_{t+1}$}
				\State $\RR_{t+1}, \dds{\RR_{t+1}} \gets \textsc{StreamSet}(t,\RR_t,\dds{\RR_t},s_{t+1})$
			\EndPasses
		\EndFor
		\Passes{2r}{4r+1}{}
			\State let $\hat{n}_r=\frac{m\cdot \dg{\RR_2}\cdot\cdots \cdot \dg{\RR_{r-1}}}{s_1\cdot \cdots\cdot s_r}$
			\StatexIndent[3] $\cdot \sum_{\vec{C}\in \RR_r} \textsc{StrIsAssigned}(\vec{C},r,\lambda,\epsilon,m,\vec{\tau})$ 
		\EndPasses
		\State \Return $\hat{n}_r$
		\EndProcedure
	\end{algorithmic}
\end{algorithm}

\begin{algorithm}
    \caption{Sample $K_{t+1}$-candidates from a stream}
    \label{alg:stream-sampleset}
    \begin{algorithmic}[1]
    \Procedure{\textsc{StreamSet}}{$t,\RR_t, \dds{\RR_t}, s_{t+1}$}
		\State set up a data structure $\mathcal{D}$ to sample each $\vec{T}\in \RR_t$
		\StatexIndent[2] with probability $\dg{\vec{T}}/\dg{\RR_t}$
		\State initialize $\RR_{t+1}=\emptyset, \dds{\RR_{t+1}}$
		\ForPar{$\ell \in [s_{t+1}]$}
			\State invoke $\mathcal{D}$ to generate $\vec{T}_\ell$
			\State $u \gets$ minimum degree vertex of $\vec{T}_\ell$
			\Pass{1}{$u, \dds{\RR_t}[u]$}
				\State query a random neighbor $w$ of $u$ \Comment{$f_3$}
			\EndPass
			\Pass{2}{$\vec{T}_\ell, w$}
				\State \textbf{if} $(\vec{T}_\ell,w)$ is a $(t+1)$-clique, \textbf{then} add it to $\RR_{t+1}$ \Comment{$f_4$}
				\State update $\dds{\RR_{t+1}}$ \Comment{$f_2$}
			\EndPass
		\EndForPar
		\State \Return $\RR_{t+1}, \dds{\RR_{t+1}}$
	\EndProcedure
	\end{algorithmic}
\end{algorithm}

\let\saveFloatBarrier\FloatBarrier%
\let\FloatBarrier\relax%
\section{Conclusion}
\let\FloatBarrier\saveFloatBarrier%
We studied the problem of estimating the number of occurrences of a subgraph $H$ in a graph $G$ in the streaming setting. We provide a transformation that converts sublinear-time graph algorithms in the query access model to sublinear-space streaming algorithms. For an arbitrary subgraph $H$, we obtained a $3$-pass algorithm in the turnstile streaming model with space complexity $\tilde{O}(\frac{m^{\rho(H)}}{\epsilon^2\cdot \# H})$ for $(1\pm \epsilon)$-approximating $\# H$, where $\rho(H)$ is the fractional edge-cover of $H$. For a clique $\# K_r$ such that $r\geq 3$, we obtained a constant-pass streaming algorithm for $(1\pm \epsilon)$-approximating $\# K_r$ in $\tilde{O}(\frac{m\lambda^{r-2}}{\epsilon^2\cdot \# K_r})$ space, in a graph $G$ with degeneracy $\lambda$. 

It would be interesting to reduce the number of passes of our algorithms even further: Can we obtain a $2$-pass algorithm for $\# H$ with space complexity $\tilde{O}(\frac{m^{\rho(H)}}{\epsilon^2\cdot \# H})$? Can we achieve a $C$-pass streaming algorithm for $\# K_r$ with space complexity $\tilde{O}(\frac{m\lambda^{r-2}}{\epsilon^2\cdot \# K_r})$ space in a graph $G$ with degeneracy $\lambda$, for some universal constant $C$ that does not depend on $r$? 

\section*{Acknowledgments}
P.P. is supported by ``the Fundamental Research Funds for the Central Universities''.

\tagged{arxiv,personal}{
	\printbibliography
}
\tagged{pods}{
\bibliographystyle{ACM-Reference-Format}
\bibliography{mybibliography}
}

\appendix
\droptag{mainpart}
\usetag{appendix}
\newpage
\section{Turnstile subgraph counting algorithm}
\label{sec:app-turnstile}

\turnstileSampler*
\turnstileSamplerProof

\droptag{io_version}
\usetag{ts_version}
\subgraphAlgorithms

\section{A sublinear-time algorithm for approximating $\# H$ of a general graph}\label{sec:sublineartimesubgraph}
Now we present the pseudo code of the subliner-time algorithm for approximating sampling and counting an arbitrary subgraph in the query access model given in \cite{fichtenberger2020sampling}. The FGP algorithm refers to \cref{subgraph-sample} (\textsc{SampleSubgraph}), which invokes two subroutines \cref{odd-cycle-sample} (\textsc{SampleOddCycle}) and  \cref{star-sample} (\textsc{SampleStar}) for sampling an odd length cycle and a star, respectively. 

\begin{algorithm}[H]
    \caption{Sampling a wedge}
    \label{wedge-sample}
    \begin{algorithmic}[1]
    \Procedure{\textsc{SampleWedge}}{$G,u,v$}
        \If{$d_u \leq \sqrt{2 m}$} \label{wedge-degree}
            \State sample a number $i \in \{ 1, \ldots \sqrt{2 m} \}$ uniformly at random
            \If{$i > d_u$}
                \State \Return \fail
            \EndIf
            \State $w$ $\leftarrow$ $i^{th}$ neighbor of $u$
        \Else
            \State sample a vertex $w$ with prob. proportional to its degree
			\State sample $t \in [0,1]$ uniformly at random
			\If{$t > \sqrt{2m} / \dg{w}$}
				\State \Return \fail
			\EndIf
        \EndIf
		\State \Return $w$
		\EndProcedure
		\end{algorithmic}
\end{algorithm}

\begin{algorithm}[H]
    \caption{Sampling a cycle of length $2k+1$}
    \label{odd-cycle-sample}
    \begin{algorithmic}[1]
    \Procedure{\textsc{SampleOddCycle}}{$G,2k+1$}  
        \State Obtain $k$ directed edges $\diedge{u_1,v_1}, \ldots, \diedge{u_k,v_k}$\StatexIndent[2] by calling \textsc{SampleEdge} $k$ times \label{odd-cycle-loop}
        \If{$u_1, v_1, \ldots, u_k,v_{k}$ is a path of length $2k-1$\StatexIndent[3] and $u_1 \prec v_1$, $\forall i > 1 : u_1 \prec u_i, v_i$} \label{path-check}
            \If{\textsc{SampleWedge}($G,u_1,v_{k}$) returns $w$ and $w \prec v_1$}
                        \State \Return $\{(u_1,v_1),\ldots,(u_k,v_k)\}\cup \{ (v_{k},w), (w,u_1) \}$
            \EndIf
        \EndIf
        \State \Return\fail
    \EndProcedure
    \end{algorithmic}
\end{algorithm}

\begin{algorithm}[H]
    \caption{Sampling a star with $k$ petals}
    \label{star-sample}
    \begin{algorithmic}[1]
        \Procedure{SampleStar}{$G,k$}
        \State Obtain $k$ directed edges $\diedge{u_1,v_1}, \ldots, \diedge{u_k,v_k}$\StatexIndent[2] by calling \textsc{SampleEdge} $k$ times
        \If{$u_1 = u_2 = \ldots = u_k$ and $v_1 \prec v_2 \prec \ldots \prec v_k$
        }
            \State \Return $(u_1, v_1, \ldots, v_k)$
        \EndIf
        \State \Return \fail
\EndProcedure
    \end{algorithmic}
\end{algorithm}

\begin{algorithm}[H]
    \caption{Sampling a copy of subgraph $H$}
    \label{subgraph-sample}
    \begin{algorithmic}[1]  
    \Procedure{SampleSubgraph}{$G, H$}
\State{Let $\overline{T}=\{\overline{C_1},\ldots,\overline{C_o},\overline{S_1},\ldots,\overline{S_s}\}$ denote\StatexIndent[2] a (decomposition) type  of $H$. }
            \ForAll{$i=1\ldots o$}
                \If{ \textsc{SampleOddCycle($G, \lvert E(\overline{C}_i) \rvert$)} returns a cycle $\mathcal{C}$}
                \State $\mathcal{C}_i \gets \mathcal{C}$\label{alg:cycle_H}
            \Else
            \State \Return \fail
            \EndIf
            \EndFor
            \ForAll{$j =1\ldots s$}
                \If{\textsc{SampleStar($G,\lvert V(\overline{S}_j) \rvert - 1$)} returns a star $\mathcal{S}$}
                    \State $\mathcal{S}_j \gets \mathcal{S}$\label{alg:star_H} 
            \Else
                \State \Return \fail
            \EndIf 
            \EndFor
                
            \State Query all edges $(\bigcup_{i \in [o]} V(\mathcal{C}_i) \cup \bigcup_{j \in [s]} V(\mathcal{S}_j))^2$
            \If{$S := (\mathcal{C}_1, \ldots, \mathcal{C}_o, \mathcal{S}_1, \ldots, \mathcal{S}_s)$ forms a copy of $H$}
                \State flip a coin and with probability $\frac{1}{f_{\overline{T}}(H)}$: \Return $S$ \label{alg:occurrence_H}
            \EndIf
            \State \Return \fail
        \EndProcedure
    \end{algorithmic}
\end{algorithm}

The authors of \cite{fichtenberger2020sampling} then make use of the FGP algorithm as a subroutine to obtain a uniform sampler of a copy of $H$ (i.e., \cref{subgraph-sample-uniformly} \textsc{SampleSubgraphUniformly}) and an estimator of $\# H$ (i.e., \cref{subgraph-estimation} \textsc{CountSubgraph}). 
\begin{algorithm}
	\caption{Sampling a copy of subgraph $H$ uniformly at random}
	\label{subgraph-sample-uniformly}
	\begin{algorithmic}[1]  
		\Procedure{SampleSubgraphUniformly}{$G, H$}
		\ForAll{$j=1,\ldots, q=10 \cdot {(2m)}^{\rho(H)}/T$}
		\State Invoke \textsc{SampleSubgraph}($G,H$)
		\If{a subgraph $H$ is returned}
		\State \Return $H$
		\EndIf
		\EndFor
		\State \Return \fail
		\EndProcedure
	\end{algorithmic}
\end{algorithm}

\begin{algorithm}
	\caption{Approximately counting the number of instances of $H$}
	\label{subgraph-estimation}
	\begin{algorithmic}[1]  
		\Procedure{CountSubgraph}{$G, H$}
		\State $X=0$
		\ForAll{$j=1,\ldots, q=10 \cdot {(2m)}^{\rho(H)}/(T\epsilon^2)$}
		\State Invoke \textsc{SampleSubgraph}($G,H$)
		\If{a subgraph $H$ is returned}
		\State $X\gets X+1$
		\EndIf
		\EndFor
		\State \Return $ X$
		\EndProcedure
	\end{algorithmic}
\end{algorithm}

\section{A sublinear-time algorithm for approximating $\# K_r$ of a graph with degeneracy at most $\lambda$}\label{app:algo_degeneracy}
This section lists pseudo code for the algorithm from \cite{BerHow20a}. The ERS-algorithm refers to \cref{clique-estimation} \textsc{CountClique}, which invokes  $\Theta(\log n)$ times \cref{basic-routine-clique-estimation} \textsc{ApproxClique} and takes the median of these outputs. 

In \cref{basic-routine-clique-estimation} \textsc{ApproxClique}, it invokes a subroutine \cref{alg:sample_a_larger_clique} for sampling a set of larger cliques, and a subroutine \cref{alg:isassigned} \textsc{IsAssigned} for checking if an ordered $r$-clique $\vec{C}$ is assigned or not. Finally, \cref{alg:isassigned} \textsc{IsAssigned} invokes \cref{alg:isactive} \textsc{IsActive} to check if all the prefixes of $\vec{C}$ is active or not. 

\begin{lemma}[\cite{eden2020faster}]\label{thm:fromerstheorem}
Let $G$ be a graph with degeneracy $\lambda$. The ERS algorithm (i.e., \cref{clique-estimation}) in the general graph model satisfies the following: 
\begin{itemize}
    \item if $L_r\in [\frac{\#K_r}{4}, \#K_r]$, then
\textsc{CountClique}($n, r,\lambda,\epsilon,m,L_r$) outputs a value $\hat{n}_r$ such that with probability at least $1-n^{-\Omega(1)}$, $\hat{n}_r$ is a $(1\pm \varepsilon)$-approximation of $\#K_r$;  
\item if $L_r > \#K_r$, then \textsc{CountClique}($n, r,\lambda,\epsilon,m,L_r$) outputs a value $\hat{n}_r$ such that with probability at least $1-n^{-\Omega(1)}$, $\hat{n}_r < L_r$;  
\item the expected running time and query complexity of the algorithm are
\[
O\left(\min\{\frac{n\lambda^{r-1}}{L_r}, \frac{n}{(\# K_r)^{1/r}}\} + \frac{m\lambda^{r-1}}{L_r}\cdot \frac{\#K_r}{L_r} \right)\cdot \mathrm{poly}(\log n, 1/\varepsilon, r^r)
\]
\end{itemize}
\end{lemma}

\begin{algorithm}[t]
	\caption{Approximately counting the number of instances of $K_r$}
	\label{clique-estimation}
	\begin{algorithmic}[1]  
		\Procedure{CountClique}{$n, r,\lambda,\epsilon,m,L_r$}
		\State $\gamma \gets \epsilon/(8r\cdot r!)$, $\beta\gets 1/(6r)$
		\State for each $t\in [2,r-1]$, set $\tau_t\gets\frac{r^{4r}}{\beta^r\cdot \gamma^2}\cdot \lambda^{r-t}$; \StatexIndent[2] $\tau_r\gets 1$; $\tau_1\gets \frac{r^{4r}}{\gamma^2}\cdot \min\{\lambda^{r-1},L_r^{(r-1)/r}\}$, \StatexIndent[2] $\vec{\tau}\gets\{\tau_1,\dots,\tau_r\}$
		\ForAll{$j=1,\ldots, q=\Theta(\log (n))$}
		\State Invoke \textsc{ApproxCliques}($n,r,\lambda,\epsilon,  L_r,m,\vec{\tau}$).%
		\State Let $\chi_j$ be the returned value. 
		\EndFor
		\State Let $\hat{n}_r$ be the the median value of $\chi_1,\cdots, \chi_q$ \State \Return $ \hat{n}_r$.
		\EndProcedure
	\end{algorithmic}
\end{algorithm}

\begin{algorithm}[t]
	\caption{Approximately counting the number of instances of $K_r$}
	\label{basic-routine-clique-estimation}
	\begin{algorithmic}[1]  
		\Procedure{ApproxClique}{$n,r,\lambda,\epsilon, L_r,m,\vec{\tau}$}
		\State set $\RR_0\gets V$, $\dg{\RR_0}\gets n$, $\tilde{\ww}_0=(1-\epsilon/2)L_r$, $\beta\gets 1/(18r)$ \StatexIndent[2] and $\gamma\gets \epsilon/(2r)$
		\State sample $s_1=\lceil \frac{n\tau_1}{\tilde{\ww}_0}\cdot \frac{3\ln(2/\beta)}{\gamma^2}\rceil$ vertices u.a.r \StatexIndent[2] and let $\RR_1$ be the chosen multiset%
		\ForAll{$t=1,\ldots, r-1$}
		\State Compute $\dg{\RR_t}$ and set $\tilde{\ww}_t=(1-\gamma)\frac{\tilde{\ww}_{t-1}}{\dg{\RR_{t-1}}}\cdot s_t$ \StatexIndent[3] and $s_{t+1}\gets \lceil \frac{\dg{\RR_t}\tau_{t+1}}{\tilde{\ww}_t}\cdot \frac{3\ln(2/\beta)}{\gamma^2}\rceil$ 
		\State If $s_{t+1}>\frac{4m\lambda^{t-1}\cdot \tau_{t+1}}{L_r}\cdot \frac{(r!)^2\cdot 3\ln (2/\beta)}{\beta^t\cdot \gamma^2}$ then \textbf{abort}
		\State Invoke \textsc{SampleASet}($t,\RR_t,s_{t+1}$) \StatexIndent[3] and let $\RR_{t+1}$ be the returned multiset. 
		\EndFor
		\State $\hat{n}_r=\frac{n\cdot \dg{\RR_1}\cdot\cdots \cdot \dg{\RR_{r-1}}}{s_1\cdot \cdots\cdot s_r}\sum_{\vec{C}\in \RR_r} \textsc{IsAssigned}(\vec{C},r,\lambda,\epsilon,L_r,m,\vec{\tau})$ 
		\State \Return $\hat{n}_r$.
		\EndProcedure
	\end{algorithmic}
\end{algorithm}

\begin{algorithm}[t]
	\caption{Sampling a set of ordered $(t+1)$-cliques}
	\label{alg:sample_a_larger_clique}
	\begin{algorithmic}[1]  
		\Procedure{SampleASet}{$t,\RR_t, s_{t+1}$}
		\State Compute $\dg{\RR_t}$ and set up a data structure \StatexIndent[2] to sample each $\vec{T}\in \RR_t$ with probability $\dg{\vec{T}}/\dg{\RR_t}$
		\State initialize $\RR_t=\emptyset$
		\ForAll{$\ell=1, \dots, s_{t+1}$}
		\State invoke the above to 
		generate $\vec{T}_\ell$
		\State find the minimum degree vertex $u$ of $\vec{T}_\ell$
		\State sample a random neighbor $w$ of $u$
		\State If the $(t+1)$-tuple $(\vec{T}_\ell,w)$ is an ordered $(t+1)$-clique, add it to $\RR_{t+1}$
		\EndFor
		\State \Return $\RR_{t+1}$.
		\EndProcedure
	\end{algorithmic}
\end{algorithm}

\begin{algorithm}[t]
	\caption{Check if an ordered $r$-clique $\vec{C}$ is assigned to the un-ordered clique}
	\label{alg:isassigned}
	\begin{algorithmic}[1]  
		\Procedure{IsAssigned}{$\vec{C},r,\lambda,\epsilon,L_r,m,\vec{\tau}$}
		\State Let $C$ be the un-ordered clique corresponding to $\vec{C}$
		\ForAll{ordered $r$-clique $\vec{C}'$ \StatexIndent[2] whose un-ordered clique equals $C$}
		\ForAll{prefix $\vec{C}_{\leq t}'$, $t\in[r-1]$}
		\State invoke \textsc{IsActive}($t,\vec{C}_{\leq t}',r,\lambda,\epsilon,L_r,m,\vec{\tau}$) \StatexIndent[4] and if it returns \textbf{Non-Active} \StatexIndent[4] then  \textbf{abort} and \Return $0$
		\EndFor
		\EndFor
		\If{$\vec{C}$ is the lexicographically first ordered $r$-clique \StatexIndent[2] in the above set of active ordered cliques}
		\State \Return $1$
		\Else
		\State  \Return $0$.
		\EndIf
		\EndProcedure
	\end{algorithmic}
\end{algorithm}

\begin{algorithm}[t]
	\caption{Check if an ordered $i$-clique $\vec{I}$ is active}
	\label{alg:isactive}
	\begin{algorithmic}[1]  
		\Procedure{IsActive}{$i,\vec{I},r,\lambda,\epsilon,L_r,m,\vec{\tau}$}
		\ForAll{$\ell=1,\cdots, q= 12\ln(n^{r+10})$}
		\State set $\RR_i=\{\vec{I}\}$, $\tilde{\ww}_i=(1-\epsilon/2)\tau_i,\beta=1/(6r)$, \StatexIndent[3] and $\gamma=\epsilon/(8r\cdot r!)$.
		\ForAll{$t=i, \dots, r-1$}
		\State Compute $\dg{\RR_t}$
		\State for $t>i$, set $\tilde{\ww}_t=(1-\gamma)\frac{\tilde{\ww}_{t-1}\cdot s_t}{\dg{\RR_{t-1}}}$ \StatexIndent[4] and $s_{t+1}=\frac{\dg{\RR_t}\cdot \tau_{t+1}}{\tilde{\ww}_t}\cdot \frac{3\ln(2/\beta)}{\gamma^2}$.
		\State if $s_{t+1}>\frac{2m\lambda^{t-1}\cdot \tau_{t+1}}{L_r}\cdot \frac{12 \ln(1/\beta)}{\beta^r\cdot\gamma^3}$, \StatexIndent[4] then set $\chi_\ell=0$ and continue to next $\ell$.
		\State invoke \textsc{SampleASet}($t,\RR_t,s_{t+1}$) \StatexIndent[4] and let $\RR_{t+1}$ be the returned multiset.
		\EndFor
		\State set $\hat{c}_r(\vec{I})=\frac{\dg{\RR_i}\cdot \cdots \cdot \dg{\RR_{r-1}}}{s_{i+1}\cdot \cdots\cdot s_r}\cdot |\RR_{r}|$.
		\State if $\hat{c}_r(\vec{I})\leq \frac{\tau_i}{4}$, then $\chi_\ell=1$, otherwise $\chi_\ell=0$.
		\EndFor
		\If{$\sum_{\ell=1}^q\chi_\ell\geq q/2$}
		\State \Return \textbf{Active}
		\Else
		\State  \Return \textbf{Non-Active}.
		\EndIf
		\EndProcedure
	\end{algorithmic}
\end{algorithm}

\section{Missing algorithms from section \ref{sec:streaming-degeneracy}}
\label{sec:app-missing}

\begin{algorithm}[H]
	\caption{Check if an ordered $r$-clique $\vec{C}$ is assigned to its unordered clique}
	\label{alg:stream-assigned}
	\begin{algorithmic}[1]  
		\Procedure{StrIsAssigned}{$\vec{C},r,\lambda,\epsilon,m,\vec{\tau}$}
			\State $C \gets$ the unordered clique corresponding to $\vec{C}$
			\State initialize dictionary $A$
			\ForPar{ordered $r$-cliques $\vec{C}'$ isomorphic to $C$}
				\ForPar{$t\in\{2,\dots,r\}$}
					\Passes{1}{2t-1}{$t,\vec{C}_{\leq t}',r,\lambda,\epsilon,m,\vec{\tau}$}
						\State $A[\vec{C}_{\leq t}'] \gets$ \textsc{StrAct}($t,\vec{C}_{\leq t}',r,\lambda,\epsilon,m,\vec{\tau}$)
					\EndPasses
				\EndForPar
			\EndForPar

			\For{ordered $r$-clique $\vec{C}'$ of $C$}
				\If{$\vec{C}'$ ``lex. <'' $\vec{C}$
					\StatexIndent[3] $\wedge \forall t \in\{2,\dots,r\}: A[\vec{C}_{\leq t}'] = \textbf{active}$}
					\State \Return 0
				\EndIf
			\EndFor
			\State \Return 1
		\EndProcedure
	\end{algorithmic}
\end{algorithm}

\begin{algorithm}[H]
	\caption{Check if an ordered $i$-clique $\vec{I}$ is active}
	\label{alg:stream-active}
	\begin{algorithmic}[1]  
		\Procedure{StrAct}{$i,\vec{I},r,\lambda,\epsilon,m,\vec{\tau}$}
		\ForPar{$\ell \gets 1, \dots, q= 12\ln(n^{r+10}/\delta)$}
			\State $\RR_i \gets \{\vec{I}\}$, $\tilde{\ww}_i \gets (1-\epsilon/2)\tau_i, \beta \gets 1/(6r)$,
			\StatexIndent[3] and $\gamma \gets \epsilon/(8r\cdot r!)$
			\Pass{$1$}{$\RR_i$}
				\State construct $\dds{\RR_i}$ \Comment{$f_2$}
			\EndPass
			\ForAll{$t \gets i, \dots, r-1$}
				\State $\dg{\RR_t} \gets \sum_{\vec{T} \in \RR_t} \dg{\vec{T}} = \sum_{\vec{T} \in \RR_t} \min_{v \in \vec{T}} \dds{\RR_t}[v]$
				\State \textbf{for} $t>i$, set $\tilde{\ww}_t=(1-\gamma)\frac{\tilde{\ww}_{t-1}\cdot s_t}{\dg{\RR_{t-1}}}$
				\StatexIndent[4] and $s_{t+1}=\frac{\dg{\RR_t}\cdot \tau_{t+1}}{\tilde{\ww}_t}\cdot \frac{3\ln(2/\beta)}{\gamma^2}$
				\If{$s_{t+1}>\frac{2m\lambda^{t-1}\cdot \tau_{t+1}}{\#K_r}\cdot \frac{12 \ln(1/\beta)}{\beta^r\cdot\gamma^3}$}
					\State $\chi_\ell \gets 0$ and break loop iteration for $\ell$
				\EndIf
				\Passes{2t}{2t+1}{$t,\RR_t,\dds{\RR_t},s_{t+1}$}
				\State $\RR_{t+1}, \dds{\RR_{t+1}}$
				\StatexIndent[5] $\gets$ \textsc{StreamSet}($t,\RR_t,\dds{\RR_t},s_{t+1}$)
				\EndPasses
			\EndFor
			\State $\hat{c}_r(\vec{I}) \gets \frac{\dg{\RR_i}\cdot \cdots \cdot \dg{\RR_{r-1}}}{s_{i+1}\cdot \cdots\cdot s_r}\cdot |\RR_{r}|$
			\State \textbf{if} $\hat{c}_r(\vec{I})\leq \frac{\tau_i}{4}$, \textbf{then} $\chi_\ell \gets 1$, \textbf{else} $\chi_\ell \gets 0$
		\EndForPar
		\If{$\sum_{\ell=1}^q\chi_\ell\geq q/2$}
		\State \Return \textbf{active}
		\Else
		\State \Return \textbf{non-active}
		\EndIf
		\EndProcedure
	\end{algorithmic}
\end{algorithm}

\end{document}